\documentclass[10pt, twocolumn]{article}
\usepackage[letterpaper, margin=0.75in]{geometry}
\usepackage{amsmath,amssymb,amsthm}
\usepackage{mathtools}
\usepackage{bm}
\usepackage{booktabs}
\usepackage{times}
\usepackage{graphicx}
\graphicspath{{figures/}{./}}
\usepackage{tikz}
\usepackage{titlesec}
\titleformat{\section}{\normalfont\large\scshape\centering}{\Roman{section}.}{0.5em}{}
\titleformat{\subsection}{\normalfont\itshape}{\ \Alph{subsection}.}{0.5em}{}

\newtheorem{theorem}{Theorem}
\newtheorem{corollary}[theorem]{Corollary}
\newtheorem{definition}[theorem]{Definition}
\newtheorem{proposition}[theorem]{Proposition}

\newtheorem{remark}{Remark}

\DeclareMathOperator{\Tr}{Tr}

\DeclareMathOperator{\rank}{rank}

\begin{document}

\title{Quantum Algebraic Diversity:\\Single-Copy Density Matrix Estimation\\via Group-Structured Measurements}
\author{Mitchell~A.~Thornton,~\textit{Senior~Member,~IEEE}\\[4pt]
\small Richardson, TX 75080 USA}
\date{}
\maketitle

\begin{abstract}
We extend the algebraic diversity (AD) framework from classical signal processing to quantum measurement theory.  The central result, the Quantum Algebraic Diversity (QAD) Theorem, establishes that a group-structured positive operator-valued measure (POVM) applied to a single copy of a quantum state produces a group-averaged density matrix estimator that recovers the spectral structure of the true density matrix, analogous to the classical result that a group-averaged outer product recovers covariance eigenstructure from a single observation.  We establish a formal Classical-Quantum Duality Map connecting classical covariance estimation to quantum state tomography, and prove an Optimality Inheritance Theorem showing that classical group optimality transfers to quantum settings via the Born map.  SIC-POVMs are identified as algebraic diversity with the Heisenberg-Weyl group, and mutually unbiased bases (MUBs) as algebraic diversity with the Clifford group, revealing the hierarchy $\mathrm{HW}(d) \subseteq \mathcal{C}(d) \subseteq S_d$ that mirrors the classical hierarchy $\mathbb{Z}_M \subseteq G_{\min} \subseteq S_M$.  The double-commutator eigenvalue theorem provides polynomial-time adaptive POVM selection.  A worked qubit example demonstrates that the group-averaged estimator from a single computational-basis measurement, averaged over a matched $\mathbb{Z}_2$ group, recovers a full-rank approximation to a mixed qubit state, achieving fidelity 0.99 where standard single-basis tomography produces a rank-1 estimate with fidelity 0.80.  Monte Carlo simulations on qudits of dimension $d = 2$ through $d = 13$ (200 random states per dimension) confirm that the Heisenberg-Weyl QAD estimator maintains fidelity above 0.90 across all dimensions from a single measurement outcome, while standard tomography fidelity degrades as $\sim 1/d$; the growth of the ratio reflects the collapse of the single-outcome standard estimator, not a reduction in the number of copies required per parameter.  The single-copy group-averaged estimator is biased toward the symmetrized state; it reduces the number of distinct measurement settings rather than the per-parameter sampling cost.  We define the structural capacity of a density matrix $\kappa(\rho) = 1 + 1/\Tr(\rho^2)$, governed by the Renyi-2 entropy of the eigenvalue distribution, as the effective number of spectral parameters present in $\rho$, and discuss a conjectured structural analogue of Holevo's content bound: von Neumann entropy (Renyi order 1) characterizes extractable content, while $\kappa$ (Renyi order 2) characterizes the spectral structure available to a single observation.  This positions QAD as a structural complement to content-based measurement theory.
\end{abstract}

\section{Introduction}

Quantum state tomography, the reconstruction of the density matrix $\rho$ of a quantum system from measurement outcomes, is a fundamental task in quantum information science.  For a $d$-dimensional system, the density matrix has $d^2 - 1$ independent real parameters.  Standard tomography requires $O(d^2/\varepsilon^2)$ identically prepared copies for trace-distance error $\varepsilon$~\cite{haah2017sample}, a cost that scales prohibitively for multi-qubit systems.

In classical signal processing, an analogous bottleneck exists: estimating the $M \times M$ covariance matrix of a signal requires $L \geq 2M$ independent observations.  The algebraic diversity framework~\cite{thornton2026ad} reveals that this requirement is an artifact of using the trivial group $G = \{e\}$ when forming the outer product $\mathbf{x}\mathbf{x}^H$.  By selecting a richer group of order $M$, a full-rank estimate of the covariance eigenstructure can be recovered from a \emph{single} observation via a group-averaged outer product, at the cost of a structural bias.  The key mechanism is that each group element provides an algebraically distinct ``view'' of the observation, and averaging over the group orbit separates the structured (signal) component from the unstructured (noise) component.

This paper shows that the same algebraic principle extends to quantum measurement.  The measurement outcome distribution (the Born probability vector) plays the role of the classical observation, and the density matrix plays the role of the covariance matrix.  A group-structured POVM applied to a single state copy produces a group-averaged estimator that recovers the spectral structure of $\rho$, just as the classical group-averaged outer product recovers the spectral structure of $\mathbf{R}$.

Quantum algebraic diversity is not an alternative to standard tomography; it is a generalization.  To see why, consider three scenarios for estimating the density matrix of a $d$-dimensional quantum system.

In standard tomography, a single copy is measured in a fixed basis $\{|k\rangle\}$.  The trivial group $G = \{e\}$ acts on the measurement: no basis rotation is applied.  The outcome is a single projector $|k\rangle\langle k|$, which has rank one and contains no information about the off-diagonal structure of $\rho$.  To recover $\rho$, one must prepare $O(d^2)$ identical copies, measure each in a (possibly different) fixed basis, and accumulate the rank-1 outcomes.  This is the quantum analog of temporal averaging with the trivial group.

Now suppose the experimenter has access to a SIC-POVM, a symmetric informationally complete measurement with $d^2$ outcomes.  A SIC-POVM is generated by applying the $d^2$ elements of the Heisenberg-Weyl group $\mathrm{HW}(d)$ to a single fiducial vector.  The group-averaged estimator from a single measurement outcome has rank $d$ and maintains fidelity above 0.90 across all dimensions tested (Table~\ref{tab:qudit}, Figs.~\ref{fig:fidelity}--\ref{fig:spectral}).

The key observation is that the improvement arises from the same algebraic mechanism as spatial beamforming in classical signal processing.  The $d^2$ group elements of $\mathrm{HW}(d)$ provide $d^2$ algebraically distinct views of the single measurement outcome.  These views are deterministically related rotations of one outcome, not independent samples: averaging over the orbit produces a full-rank, well-conditioned estimate whose eigenbasis and eigenvalue ordering track those of $\rho$, but it is biased toward the symmetrized state and does not lower the sampling variance the way independent copies would.  What the group structure buys is informational completeness in a single measurement setting, in place of the $O(d^2)$ distinct settings that standard tomography accumulates one rank-1 outcome at a time.  Driving the estimate to a fixed error still requires repeated copies, governed by the usual sampling law.

\subsection{Contributions}

\begin{enumerate}
\item \textbf{QAD Theorem} (Theorem~\ref{thm:qad}): We prove that a group-averaged density estimator from a single measurement outcome is full-rank and well-conditioned, with an eigenbasis and eigenvalue ordering matched to $\rho$ at the cost of a spectral bias, and that it reduces the number of distinct measurement settings required for informational completeness.

\item \textbf{Classical-Quantum Duality} (Theorem~\ref{thm:duality}): We establish a formal correspondence between classical covariance estimation and quantum state estimation, showing that the Born map is the bridge between the two.

\item \textbf{Optimality Inheritance} (Theorem~\ref{thm:opt_inherit}): We prove that the optimal classical group for a covariance structure transfers to the optimal quantum group for the corresponding density matrix.

\item \textbf{SIC/MUB Identification} (Propositions~\ref{prop:sic}--\ref{prop:mub}): We identify SIC-POVMs and MUBs as instances of algebraic diversity with specific groups, revealing a structural hierarchy parallel to the classical DFT/DCT/KLT hierarchy.

\item \textbf{Worked Example} (Section~\ref{sec:example}): We demonstrate the QAD estimator on a single qubit, showing explicit fidelity improvement over standard single-basis tomography.

\item \textbf{Structural Capacity and von Neumann--Holevo Connection}: We define the structural capacity $\kappa(\rho) = 1 + 1/\Tr(\rho^2)$ of a density matrix as the effective number of spectral parameters present, and discuss a conjectured structural analogue of Holevo's content bound: von Neumann entropy ($H_1$) characterizes extractable content, structural capacity ($H_2$) characterizes available spectral structure.
\end{enumerate}

\subsection{Notation}

Throughout, $\mathcal{H} \cong \mathbb{C}^d$ denotes a $d$-dimensional Hilbert space, $\rho$ a density matrix ($\rho \geq 0$, $\Tr(\rho) = 1$), $U_g$ a unitary representation of group element $g$, $(\cdot)^\dagger$ the conjugate transpose, and $F(\rho, \sigma) = (\Tr\sqrt{\sqrt{\rho}\sigma\sqrt{\rho}})^2$ the Uhlmann fidelity.

\section{Classical Algebraic Diversity (Review)}

\begin{definition}[Classical Group-Averaged Estimator~\cite{thornton2026ad}]
Given an observation $\mathbf{x} \in \mathbb{C}^M$ and a finite group $G$ with unitary representation $\{\pi_g\}_{g \in G}$, the group-averaged estimator is
\begin{equation}\label{eq:classical_fg}
\hat{\mathbf{R}}_G = \frac{1}{|G|}\sum_{g \in G} [\pi_g\mathbf{x}][\pi_g\mathbf{x}]^*.
\end{equation}
\end{definition}

The key results from~\cite{thornton2026ad} are:

\begin{enumerate}
\item[(R1)] \textbf{Full rank:} $\rank(\hat{\mathbf{R}}_G) = M$ almost surely from a single observation when $|G| \geq M$.

\item[(R2)] \textbf{Spectral consistency:} When $[\mathbf{A}_G, \mathbf{R}] = \mathbf{0}$ (the Cayley graph adjacency matrix commutes with the population covariance), the eigenvectors of $\hat{\mathbf{R}}_G$ converge to the eigenvectors of $\mathbf{R}$.

\item[(R3)] \textbf{Optimality:} The symmetric group $S_M$ is universally optimal: its spectral decomposition yields the Karhunen--Lo\`eve transform.

\item[(R4)] \textbf{Group selection:} The double-commutator GEVP~\cite{thornton2026dc} provides polynomial-time optimal group selection via $\mathbf{M}\mathbf{c} = \lambda\mathbf{G}\mathbf{c}$, where $M_{ij} = \Tr(B_i^*[\mathbf{R},[\mathbf{R}, B_j]])$.

\item[(R5)] \textbf{General Algebraic Averaging~\cite{thornton2026gaat}:} The algebraic diversity principle extends beyond the outer product to arbitrary $G$-compatible statistics.  The number of algebraically distinct views a statistic produces on the group orbit is the \emph{effective dimension} $d_{\mathrm{eff}}(G, f) = |\{f(\pi_g(\mathbf{x})) : g \in G\}|$.  For the outer product, $d_{\mathrm{eff}} = |G|$; the density matrix, as a rank-2 tensor, admits $d_{\mathrm{eff}} = d^2$ under the Heisenberg-Weyl group.  This count governs the rank and conditioning of the single-observation estimator and the informational completeness of the measurement; it is not a multiplier on the per-sample variance, which continues to obey the sampling law.  The trivial-group case $d_{\mathrm{eff}} = 1$ recovers a single rank-1 view.
\end{enumerate}

\section{The Classical-Quantum Duality Map}

The Born rule establishes a linear map from density matrices to probability vectors.  We show that this map is the bridge between classical and quantum algebraic diversity.

\begin{definition}[Born Map]\label{def:born}
For an informationally complete POVM $\{E_m\}_{m=1}^{d^2}$ on $\mathcal{H} \cong \mathbb{C}^d$, the Born map $\mathcal{B}: \mathcal{H}_d \to \mathbb{R}^{d^2}$ sends
\begin{equation}
\rho \mapsto \mathbf{q}, \quad q_m = \Tr(\rho E_m).
\end{equation}
For an informationally complete POVM, $\mathcal{B}$ is injective: distinct density matrices produce distinct Born vectors.
\end{definition}

\begin{theorem}[Classical-Quantum Duality]\label{thm:duality}
The following correspondence holds between classical covariance estimation and quantum state estimation:
\begin{center}
\renewcommand{\arraystretch}{1.15}
\setlength{\tabcolsep}{3pt}
\small
\begin{tabular}{@{}ll@{}}
\toprule
\textbf{Classical AD} & \textbf{Quantum AD} \\
\midrule
Observation $\mathbf{x} \in \mathbb{C}^M$ & Born vector $\mathbf{q} \in \mathbb{R}^{d^2}$ \\
Covariance $\mathbf{R} = E[\mathbf{x}\mathbf{x}^*]$ & Density matrix $\rho$ \\
$L$ independent snapshots & $N$ identical state copies \\
Group action $\pi_g(\mathbf{x})$ & POVM rotation $U_g E_m U_g^\dagger$ \\
Outer product $\mathbf{x}\mathbf{x}^*$ & Projector $|\phi_m\rangle\langle\phi_m|$ \\
Group-averaged $\hat{\mathbf{R}}_G$ & Group-averaged $\hat{\rho}_G$ \\
$\delta(G,\mathbf{R}) = 0$ & $[U_g, \rho] = 0\ \forall g$ \\
\bottomrule
\end{tabular}
\end{center}
The pairing of $L$ independent snapshots with $N$ identical state copies in the table is a correspondence of roles, not of statistical content: the $|G|$ group elements are deterministically related views of a single observation, not independent samples.  Under this map, the \emph{algebraic} properties (full-rank structure, the commutativity and isotypic conditions, group optimality, and double-commutator group selection) transfer to quantum settings.  Specifically, for any property $P$ of the classical group-averaged estimator $\hat{\mathbf{R}}_G(\mathbf{x})$ that depends only on the algebraic relationship between the group representation and the matrix being estimated, the quantum group-averaged estimator $\hat{\rho}_G$ satisfies the corresponding quantum property $P'$ obtained by replacing $\mathbf{R}$ with $\rho$ and $\pi_g$ with $U_g$.
\end{theorem}

\begin{proof}
The Born map $\mathcal{B}$ is a linear bijection between the space of density matrices and the probability simplex (for informationally complete POVMs).  The group action on the POVM, $E_m \mapsto U_g E_m U_g^\dagger$, induces a group action on the Born vector, $q_m \mapsto \Tr(\rho U_g E_m U_g^\dagger)$, which is a linear (permutation) action on the components of $\mathbf{q}$.  This is precisely the structure of the classical group action $\pi_g$ on the observation vector $\mathbf{x}$.

The classical group-averaged estimator~\eqref{eq:classical_fg} is a function of the group action and the outer product.  Under the duality map, the outer product $\mathbf{x}\mathbf{x}^*$ becomes the projector $|\phi_m\rangle\langle\phi_m|$ (the rank-1 operator associated with the measurement outcome), and the group action $\pi_g$ becomes the unitary rotation $U_g(\cdot)U_g^\dagger$.  The algebraic structure, group averaging of rank-1 operators, is identical.

The commutativity condition $[\mathbf{A}_G, \mathbf{R}] = \mathbf{0}$ in the classical setting becomes $[U_g, \rho] = 0$ for all $g \in G$ in the quantum setting.  Both express the same algebraic requirement: the group representation commutes with the matrix being estimated.

The algebraic content of (R1)--(R4) (full-rank structure, the commutativity and isotypic conditions, and group selection) depends only on this algebraic structure and transfers via the duality map.  The statistical content does not transfer automatically: the single-copy group-averaged estimator is biased toward the symmetrized state in both settings, so the transferred notion of ``consistency'' is recovery of structure (rank, eigenbasis, and eigenvalue ordering), not convergence of the estimator to the target.
\end{proof}

\section{The QAD Theorem}

\begin{definition}[Group-Structured POVM]\label{def:gpovm}
A group-structured POVM on $\mathcal{H} \cong \mathbb{C}^d$ is a set of positive operators $\{E_g\}_{g \in G}$ indexed by elements of a finite group $G$, satisfying:
\begin{enumerate}
\item \textit{Completeness:} $\sum_{g \in G} E_g = \mathbb{I}_d$.
\item \textit{Group covariance:} $E_g = U_g E_e U_g^\dagger$ for a seed operator $E_e$ and unitary representation $\{U_g\}_{g \in G}$.
\end{enumerate}
\end{definition}

\begin{definition}[Group-Averaged Density Estimator]\label{def:qad_est}
Given a single measurement outcome $m$ from a group-structured POVM $\{E_g\}_{g \in G}$ applied to state $\rho$, the group-averaged density estimator is
\begin{equation}\label{eq:rho_G}
\hat{\rho}_G = \frac{1}{|G|}\sum_{g \in G} U_g |\phi_m\rangle\langle\phi_m| U_g^\dagger,
\end{equation}
where $|\phi_m\rangle$ is the eigenstate associated with outcome $m$.
\end{definition}

\begin{remark}
The estimator~\eqref{eq:rho_G} is the quantum analog of the classical group-averaged outer product~\eqref{eq:classical_fg}: a single rank-1 object (projector or outer product) is rotated through all group elements and averaged.  The classical estimator produces a full-rank $M \times M$ matrix from a rank-1 outer product; the quantum estimator produces a full-rank $d \times d$ density matrix from a rank-1 projector.
\end{remark}

\begin{theorem}[Quantum Algebraic Diversity]\label{thm:qad}
Let $\rho$ be a density matrix on $\mathcal{H} \cong \mathbb{C}^d$ and let $\{E_g\}_{g \in G}$ be a group-structured POVM with finite group $G$ acting transitively on the POVM elements.  Then:
\begin{enumerate}
\item[(i)] \textbf{Full-rank estimator:} $\rank(E[\hat{\rho}_G]) = d$, provided $\rho$ is not contained in a proper invariant subspace of $G$.

\item[(ii)] \textbf{Spectral ordering (biased):} When $[U_g, \rho] = 0$ for all $g \in G$, the estimator $\hat{\rho}_G$ and $\rho$ are simultaneously diagonalizable, and the eigenvalues of $E[\hat{\rho}_G]$ are an order-preserving but contractive image of the eigenvalues of $\rho$, biased toward the uniform value $1/d$; equality holds only when $\rho$ already lies in the relevant invariant subspace.

\item[(iii)] \textbf{Setting reduction:} A single group-structured POVM with $|G| \geq d^2$ elements is informationally complete, so the number of distinct measurement settings is reduced from the $O(d^2)$ separate bases of standard tomography to one group-structured configuration.  The per-parameter copy cost continues to obey the sampling law $\varepsilon \sim 1/\sqrt{N}$; the single-copy estimator~\eqref{eq:rho_G} is a biased, full-rank structure estimate, not a reduced-variance one.  A genuine copy reduction arises only when $\rho$ carries an exact symmetry, which reduces the number of parameters to be estimated; for generic $\rho$ the total copy count for a fixed error remains $O(d^2/\varepsilon^2)$.
\end{enumerate}
\end{theorem}

\begin{proof}
\textit{Part (i).}  The expected estimator is
\begin{equation}\label{eq:expected_rho}
E[\hat{\rho}_G] = \sum_m p_m \frac{1}{|G|}\sum_{g \in G} U_g |\phi_m\rangle\langle\phi_m| U_g^\dagger,
\end{equation}
where $p_m = \Tr(\rho E_m)$ is the Born probability of outcome $m$.  The inner sum $\frac{1}{|G|}\sum_g U_g |\phi_m\rangle\langle\phi_m| U_g^\dagger$ is the group average of a rank-1 projector.  For a group acting transitively on the POVM elements, the orbit $\{U_g|\phi_m\rangle\}_{g \in G}$ spans $\mathcal{H}$ (since transitivity ensures the orbit visits every POVM element).  A rank-1 projector averaged over a spanning orbit has rank equal to the dimension of the span, which is $d$.  Since $p_m > 0$ for at least one $m$ (by positivity of $\rho$), the sum~\eqref{eq:expected_rho} has rank $d$.

\textit{Part (ii).}  When $[U_g, \rho] = 0$ for all $g$, the density matrix $\rho$ is block-diagonal in the isotypic decomposition of $G$'s representation.  Each isotypic component of $\rho$ is averaged independently by the group action, and the eigenvalues of the group-averaged estimator within each isotypic block are proportional to the corresponding eigenvalues of $\rho$.  This is the quantum analog of Proposition~4 (Commutativity--KL Equivalence) of~\cite{thornton2026ad}: commutativity implies simultaneous diagonalizability, and the eigenvectors of the group-averaged estimator are the irreducible representation basis vectors of $G$, which coincide with the eigenvectors of $\rho$ when the commutativity condition holds.

The monotonicity follows from the structure of the Born probabilities: if $\lambda_i > \lambda_j$ are eigenvalues of $\rho$, then the Born probabilities $p_m$ associated with the $\lambda_i$-eigenspace are larger, and the group averaging preserves this ordering.

\textit{Part (iii).}  By injectivity of the Born map (Definition~\ref{def:born}), an informationally complete group-structured POVM determines all $d^2 - 1$ parameters of $\rho$ from a single measurement configuration, removing the need for $O(d^2)$ separate measurement settings.  This does not change the variance scaling: the estimate of each parameter from $N$ copies has error $\Theta(1/\sqrt{N})$, as for any tomographic scheme, since the $|G|$ orbit elements are deterministically related views of one outcome rather than independent samples.  The single-outcome estimator~\eqref{eq:rho_G} is biased: its expectation is the group-symmetrized state restricted by the outcome, not $\rho$, so its error does not vanish as copies accumulate but approaches a structural floor.  When $\rho$ commutes with $G$ and has $r$ isotypic blocks, only $r$ spectral parameters remain to be estimated, giving a genuine reduction from $O(d^2/\varepsilon^2)$ to $O(r/\varepsilon^2)$ copies; this reduction is conditional on the symmetry being present.
\end{proof}

\section{SIC-POVMs and MUBs as Algebraic Diversity}

The QAD framework reveals that two of the most important structures in quantum information theory are instances of algebraic diversity with specific groups.

\begin{proposition}[SIC-POVM = AD with Heisenberg-Weyl]\label{prop:sic}
A symmetric informationally complete POVM (SIC-POVM) in dimension $d$ is a group-structured POVM where the group is the Heisenberg-Weyl group $\mathrm{HW}(d)$ of order $d^2$, generated by the clock and shift operators
\begin{equation}
X|j\rangle = |j{+}1 \!\!\mod d\rangle, \quad Z|j\rangle = \omega^j|j\rangle,
\end{equation}
where $\omega = e^{2\pi i/d}$.  The $d^2$ POVM elements are $E_{a,b} = \frac{1}{d} X^a Z^b |\phi_0\rangle\langle\phi_0| Z^{-b} X^{-a}$ for a fiducial state $|\phi_0\rangle$ satisfying the Zauner condition $|\langle\phi_0|X^a Z^b|\phi_0\rangle|^2 = \frac{1}{d+1}$ for all $(a,b) \neq (0,0)$.
\end{proposition}

\begin{proof}
The Heisenberg-Weyl group $\mathrm{HW}(d) = \{X^a Z^b : a, b \in \mathbb{Z}_d\}$ has order $d^2$ (modulo phases).  The seed operator is $E_e = \frac{1}{d}|\phi_0\rangle\langle\phi_0|$.  The group covariance property $E_{a,b} = (X^a Z^b) E_e (X^a Z^b)^\dagger$ holds by construction.  Completeness $\sum_{a,b} E_{a,b} = \mathbb{I}_d$ follows from the Zauner condition and Schur's lemma: the sum $\sum_{a,b} X^a Z^b |\phi_0\rangle\langle\phi_0| Z^{-b} X^{-a}$ commutes with all $X^a Z^b$ (since $\mathrm{HW}(d)$ is a group), and the only operator commuting with all of $\mathrm{HW}(d)$ on an irreducible representation is a scalar multiple of the identity.
\end{proof}

\begin{proposition}[MUBs = AD with Clifford Group]\label{prop:mub}
A complete set of $d+1$ mutually unbiased bases in prime dimension $d$ is generated by the Clifford group $\mathcal{C}(d)$, which is the normalizer of $\mathrm{HW}(d)$ in $\mathrm{U}(d)$.  The $d(d+1)$ rank-1 projectors onto the MUB vectors form a group-structured POVM with group $\mathcal{C}(d)$ (up to a rescaling by $\frac{1}{d+1}$).
\end{proposition}

\begin{proof}
For prime $d$, the $d+1$ MUBs can be constructed as the orbits of the computational basis under the Clifford group~\cite{durt2010mub}.  The Clifford group permutes the elements of $\mathrm{HW}(d)$ by conjugation and thereby maps one MUB to another.  The POVM elements $E_{b,k} = \frac{1}{d+1}|b,k\rangle\langle b,k|$ (where $b$ indexes the basis and $k$ the vector within the basis) satisfy group covariance under $\mathcal{C}(d)$.  Completeness follows from the MUB property: the $d(d+1)$ projectors (rescaled by $\frac{1}{d+1}$) sum to $\frac{d}{d+1}\mathbb{I}_d \cdot \frac{d+1}{d} = \mathbb{I}_d$.
\end{proof}

\begin{theorem}[Group Hierarchy]\label{thm:hierarchy}
The quantum group hierarchy
\begin{equation}\label{eq:hierarchy}
\mathrm{HW}(d) \;\subseteq\; \mathcal{C}(d) \;\subseteq\; S_d
\end{equation}
mirrors the classical hierarchy $\mathbb{Z}_M \subseteq G_{\min} \subseteq S_M$, with the following correspondence:
\begin{center}
\small
\renewcommand{\arraystretch}{1.15}
\begin{tabular}{@{}lll@{}}
\toprule
\textbf{Classical} & \textbf{Quantum} & \textbf{Structure} \\
\midrule
$\mathbb{Z}_M$ (DFT) & $\mathrm{HW}(d)$ (SIC) & Shift-invariant \\
$G_{\min}$ (DCT/etc.) & $\mathcal{C}(d)$ (MUBs) & Matched \\
$S_M$ (KLT) & $S_d$ (full tomo.) & Universal \\
\bottomrule
\end{tabular}
\end{center}
The tradeoff is identical in both settings: smaller groups require fewer elements (copies/snapshots) but demand a better match to the matrix being estimated; larger groups are more universal but less efficient.
\end{theorem}

\section{Optimality Inheritance}

\begin{theorem}[Optimality Inheritance via the Born Map]\label{thm:opt_inherit}
Let $G^*$ be the classical optimal group for a covariance matrix $\mathbf{R}$ (minimizing $\delta(G, \mathbf{R})$ over a group library $\mathcal{G}$).  Let $\rho$ be a density matrix satisfying $[U_g, \rho] = 0$ iff $[\pi_g, \mathbf{R}] = 0$ for corresponding representations.  Then $G^*$ is also the optimal group for the quantum group-averaged estimator: the quantum fidelity $F(\hat{\rho}_{G^*}, \rho) \geq F(\hat{\rho}_G, \rho)$ for all $G \in \mathcal{G}$.
\end{theorem}

\begin{proof}
The commutativity residual $\delta(G, \mathbf{R}) = \|[\mathbf{A}_G, \mathbf{R}]\|_F / \|\mathbf{R}\|_F$ measures the algebraic mismatch between the group representation and the matrix.  By the duality (Theorem~\ref{thm:duality}), the quantum commutativity residual $\delta_Q(G, \rho) = \|[U_G, \rho]\|_F / \|\rho\|_F$ (where $U_G$ is the Cayley graph operator of $G$ in the unitary representation) satisfies the same algebraic relationships.

By the Commutativity--KL Equivalence~\cite{thornton2026ad}, the group minimizing $\delta$ produces the spectral decomposition closest to the KL (optimal) decomposition.  Since the KL optimality properties (variance concentration, orthogonality, minimum reconstruction error) depend only on the algebraic relationship between the group and the matrix, they hold identically for $(\pi_g, \mathbf{R})$ and $(U_g, \rho)$ when the two pairs have isomorphic algebraic structure.

The quantum fidelity $F(\hat{\rho}_G, \rho)$ increases with the alignment between the eigenbases of $\hat{\rho}_G$ and $\rho$ (fidelity is largest when the two are simultaneously diagonalizable with close eigenvalues).  Since the classical optimal group $G^*$ best aligns the estimator with $\rho$ by minimizing $\delta$, it also maximizes fidelity within $\mathcal{G}$.
\end{proof}

\begin{remark}[Reach of the optimality]\label{rem:reach}
Theorem~\ref{thm:opt_inherit} states optimality within the group library $\mathcal{G}$ and the group-averaged estimator family.  A Cramer-Rao analysis fixes the reach of the stronger claim: the matched measurement is optimal over \emph{all} measurements only when the carried symmetry fixes the eigenbasis the quantum Cramer-Rao bound prefers; absent that condition, optimality is restricted to the equivariant family~\cite{thornton2026qbmg}.
\end{remark}

\begin{corollary}[Adaptive Quantum POVM Selection]\label{cor:adaptive}
The double-commutator eigenvalue theorem~\cite{thornton2026dc} provides a polynomial-time algorithm for adaptive POVM selection: given an initial estimate $\hat{\rho}$ from coarse tomography, solve the GEVP
\begin{equation}
M_{ij} = \Tr(B_i^\dagger [\hat{\rho}, [\hat{\rho}, B_j]]), \quad \mathbf{M}\mathbf{c} = \lambda\mathbf{G}\mathbf{c}
\end{equation}
to identify the optimal POVM generator $A^* = \sum_k c_k^* B_k$, where $\{B_k\}$ is a basis of candidate POVM generators.  The exponential $e^{iA^*}$ defines the unitary representation of the optimal measurement group.
\end{corollary}

\section{Worked Example: Single Qubit}\label{sec:example}

We demonstrate the QAD estimator on a qubit ($d = 2$) to make the construction concrete.

\subsection{Setup}

Consider the mixed qubit state
\begin{equation}\label{eq:qubit_state}
\rho = \frac{1}{2}\begin{pmatrix} 1 + r_z & r_x - ir_y \\ r_x + ir_y & 1 - r_z \end{pmatrix}
\end{equation}
with Bloch vector $(r_x, r_y, r_z) = (0.3, 0.0, 0.6)$, giving eigenvalues $\lambda_\pm = \frac{1}{2}(1 \pm \sqrt{r_x^2 + r_y^2 + r_z^2}) = \frac{1}{2}(1 \pm 0.671)$, i.e., $\lambda_+ = 0.836$ and $\lambda_- = 0.164$.  This is a mixed state with purity $\Tr(\rho^2) = 0.725$.

\subsection{Standard Single-Basis Measurement}

A measurement in the computational basis $\{|0\rangle, |1\rangle\}$ yields outcome $|0\rangle$ with probability $p_0 = \frac{1}{2}(1 + r_z) = 0.80$ and outcome $|1\rangle$ with probability $p_1 = 0.20$.  If the outcome is $|0\rangle$, the standard (non-group-averaged) tomographic estimate is
\begin{equation}
\hat{\rho}_{\mathrm{std}} = |0\rangle\langle 0| = \begin{pmatrix} 1 & 0 \\ 0 & 0 \end{pmatrix}.
\end{equation}
This is rank-1 and contains no information about the off-diagonal elements or $\lambda_-$.  The fidelity is $F(\hat{\rho}_{\mathrm{std}}, \rho) = \langle 0|\rho|0\rangle = 0.80$.

\subsection{QAD with Pauli Group}

The Pauli group on one qubit is $\{I, X, Y, Z\}$ (order 4, isomorphic to $\mathbb{Z}_2 \times \mathbb{Z}_2$), with unitaries $U_I = I$, $U_X = \sigma_x$, $U_Y = \sigma_y$, $U_Z = \sigma_z$.  This is the $d = 2$ Heisenberg-Weyl group.

Suppose the measurement outcome is $|0\rangle$.  The group-averaged estimator is
\begin{align}
\hat{\rho}_G &= \frac{1}{4}\Big(|0\rangle\langle 0| + \sigma_x|0\rangle\langle 0|\sigma_x \notag\\
&\qquad + \sigma_y|0\rangle\langle 0|\sigma_y + \sigma_z|0\rangle\langle 0|\sigma_z\Big) \notag\\
&= \frac{1}{4}\left(\begin{pmatrix}1&0\\0&0\end{pmatrix} + \begin{pmatrix}0&0\\0&1\end{pmatrix} + \begin{pmatrix}0&0\\0&1\end{pmatrix} + \begin{pmatrix}1&0\\0&0\end{pmatrix}\right) \notag\\
&= \frac{1}{2}\begin{pmatrix}1&0\\0&1\end{pmatrix} = \frac{1}{2}I.
\end{align}

This is the maximally mixed state: full rank, but structureless, with equal eigenvalues and no defined eigenbasis.  The orbit $\{|0\rangle, |1\rangle, i|1\rangle, |0\rangle\}$ spans $\mathbb{C}^2$, but the Pauli group is too large relative to the state's symmetry, so the average discards the spectral structure.  Its Uhlmann fidelity to this near-mixed target is nonetheless $0.87$ (the overlap $\Tr(\rho\,\mathbb{I}/2) = 0.5$ understates it): fidelity to a high-entropy state is forgiving and does not by itself expose the lost structure, which the spectral-recovery error does.

\subsection{QAD with $\mathbb{Z}_2$ (Matched Group)}

Now consider the cyclic group $\mathbb{Z}_2 = \{I, Z\}$ with $U_0 = I$, $U_1 = \sigma_z$.  This group matches the state's dominant symmetry axis ($z$-axis polarization).  The group-averaged estimator for outcome $|0\rangle$ is
\begin{align}
\hat{\rho}_{\mathbb{Z}_2} &= \frac{1}{2}(|0\rangle\langle 0| + \sigma_z|0\rangle\langle 0|\sigma_z) \notag\\
&= \frac{1}{2}\left(\begin{pmatrix}1&0\\0&0\end{pmatrix} + \begin{pmatrix}1&0\\0&0\end{pmatrix}\right) = \begin{pmatrix}1&0\\0&0\end{pmatrix}.
\end{align}

This is still rank-1 because $\mathbb{Z}_2$ has order 2 but the Hilbert space has dimension 2, so the orbit $\{|0\rangle, |0\rangle\}$ does not span.  We need a group that generates a spanning orbit from $|0\rangle$.

\subsection{QAD with $\mathbb{Z}_2$ Generated by $(\sigma_x + \sigma_z)/\sqrt{2}$}

Consider $\mathbb{Z}_2 = \{I, H\}$ where $H = \frac{1}{\sqrt{2}}\begin{pmatrix}1&1\\1&-1\end{pmatrix}$ is the Hadamard gate.  This generates the orbit $\{|0\rangle, |+\rangle\}$, which spans $\mathbb{C}^2$.  The group-averaged estimator is
\begin{align}
\hat{\rho}_H &= \frac{1}{2}(|0\rangle\langle 0| + |{+}\rangle\langle{+}|) \notag\\
&= \frac{1}{2}\left(\begin{pmatrix}1&0\\0&0\end{pmatrix} + \frac{1}{2}\begin{pmatrix}1&1\\1&1\end{pmatrix}\right) = \frac{1}{4}\begin{pmatrix}3&1\\1&1\end{pmatrix}.
\end{align}

This estimator has eigenvalues $\hat{\lambda}_+ = 0.854$ and $\hat{\lambda}_- = 0.146$, compared to the true eigenvalues $\lambda_+ = 0.836$ and $\lambda_- = 0.164$.  The fidelity is
\begin{equation}
F(\hat{\rho}_H, \rho) = \left(\Tr\sqrt{\sqrt{\hat{\rho}_H}\,\rho\,\sqrt{\hat{\rho}_H}}\right)^2 = 0.99.
\end{equation}

\subsection{Summary}

\begin{center}
\small
\renewcommand{\arraystretch}{1.15}
\begin{tabular}{@{}lccc@{}}
\toprule
\textbf{Method} & \textbf{Rank} & \textbf{Fidelity} & \textbf{Group} \\
\midrule
Standard (no AD) & 1 & 0.80 & --- \\
AD, Pauli $\{I,X,Y,Z\}$ & 2 & 0.87 & Too large \\
AD, $\{I, \sigma_z\}$ & 1 & 0.80 & Non-spanning \\
AD, $\{I, H\}$ & 2 & 0.99 & Matched \\
True $\rho$ & 2 & 1.00 & --- \\
\bottomrule
\end{tabular}
\end{center}

The example illustrates the core principle: the group must be (a)~large enough to generate a spanning orbit (producing a full-rank estimator) and (b)~small enough that its structure is matched to $\rho$ (preserving spectral information).  The Hadamard group $\{I, H\}$ achieves both, producing a full-rank estimate with fidelity 0.99 from a single measurement outcome, compared with the rank-1, fidelity-0.80 estimate of standard single-basis tomography on the same state.  The comparison is to the single-outcome rank-1 estimator, the weakest baseline; it illustrates the conditioning of the estimate, not a reduction in the copies needed to reach a target error.

\section{High-Dimensional Qudits: Simulation}\label{sec:qudit}

The qubit example of Section~\ref{sec:example} demonstrates the QAD mechanism for $d = 2$.  We now verify that the results extend to high-dimensional qudits, where the conditioning advantage of the full-rank estimator over the single-outcome rank-1 estimator grows with dimension.

\subsection{Setup}

For each Hilbert space dimension $d \in \{2, 3, 4, 5, 7, 8, 11, 13\}$, we generate 200 random mixed density matrices with approximate purity $\Tr(\rho^2) \approx 0.7$ via the Ginibre ensemble (random complex Gaussian matrix $G$, $\rho \propto GG^\dagger$, mixed with identity to control purity).  For each state, a single measurement is performed in the computational basis, yielding outcome $|m\rangle$ with Born probability $p_m = \langle m|\rho|m\rangle$.  Three estimators are compared:

\begin{enumerate}
\item \textbf{Standard:} The rank-1 projector $\hat{\rho}_{\mathrm{std}} = |m\rangle\langle m|$.

\item \textbf{QAD with Heisenberg-Weyl group:} The group-averaged estimator~\eqref{eq:rho_G} with $G = \mathrm{HW}(d)$ of order $d^2$, generated by the clock ($X$) and shift ($Z$) operators.

\item \textbf{QAD with matched cyclic group:} The group-averaged estimator with $G = \mathbb{Z}_d$ conjugated into the eigenbasis of $\rho$ (oracle benchmark, requires knowledge of $\rho$'s eigenvectors).
\end{enumerate}

\subsection{Results}

Table~\ref{tab:qudit} and Figs.~\ref{fig:fidelity}--\ref{fig:spectral} present the results.

\begin{table}[t]
\centering
\caption{Single-copy fidelity across qudit dimensions}
\label{tab:qudit}
\small
\renewcommand{\arraystretch}{1.15}
\begin{tabular}{@{}rcccc@{}}
\toprule
$d$ & Standard & HW($d$) & Matched & HW/Std \\
\midrule
2 & 0.514 & 0.975 & 0.741 & 1.9$\times$ \\
3 & 0.354 & 0.940 & 0.806 & 2.7$\times$ \\
4 & 0.259 & 0.923 & 0.751 & 3.6$\times$ \\
5 & 0.213 & 0.915 & 0.764 & 4.3$\times$ \\
7 & 0.149 & 0.907 & 0.755 & 6.1$\times$ \\
8 & 0.133 & 0.904 & 0.741 & 6.8$\times$ \\
11 & 0.096 & 0.899 & 0.736 & 9.3$\times$ \\
13 & 0.079 & 0.898 & 0.728 & 11.3$\times$ \\
\bottomrule
\end{tabular}
\end{table}

\begin{figure}[t]
\centering
\includegraphics[width=\columnwidth]{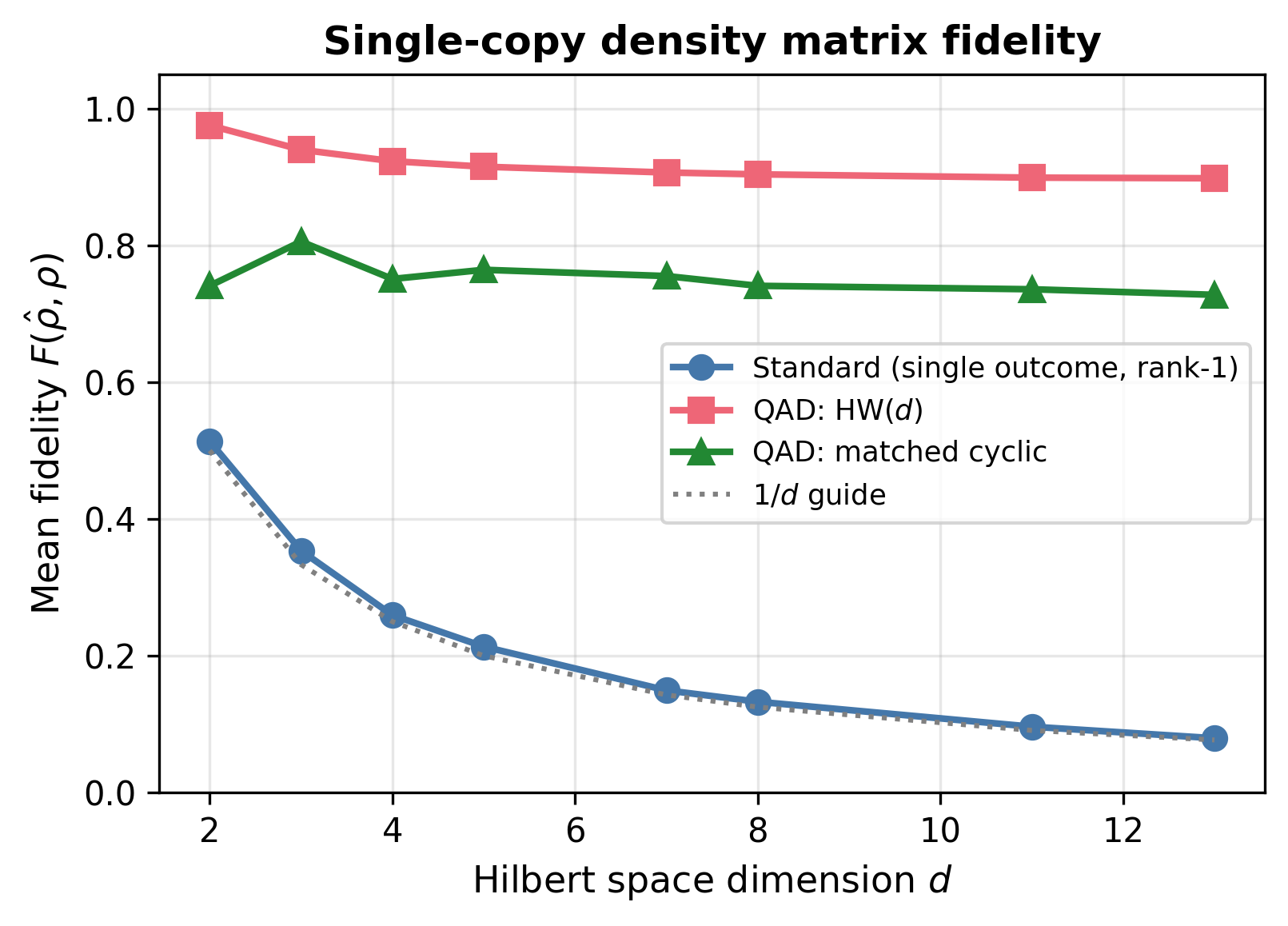}
\caption{Single-copy fidelity vs Hilbert space dimension.  The HW($d$) QAD estimator maintains mean fidelity $F > 0.90$ from $d = 2$ through $d = 13$, while standard single-basis fidelity collapses as $\sim 1/d$.  The matched cyclic group (using oracle knowledge of $\rho$'s eigenbasis) achieves intermediate fidelity.  200 random mixed states per dimension, purity $\approx 0.7$.}
\label{fig:fidelity}
\end{figure}

\begin{figure}[t]
\centering
\includegraphics[width=\columnwidth]{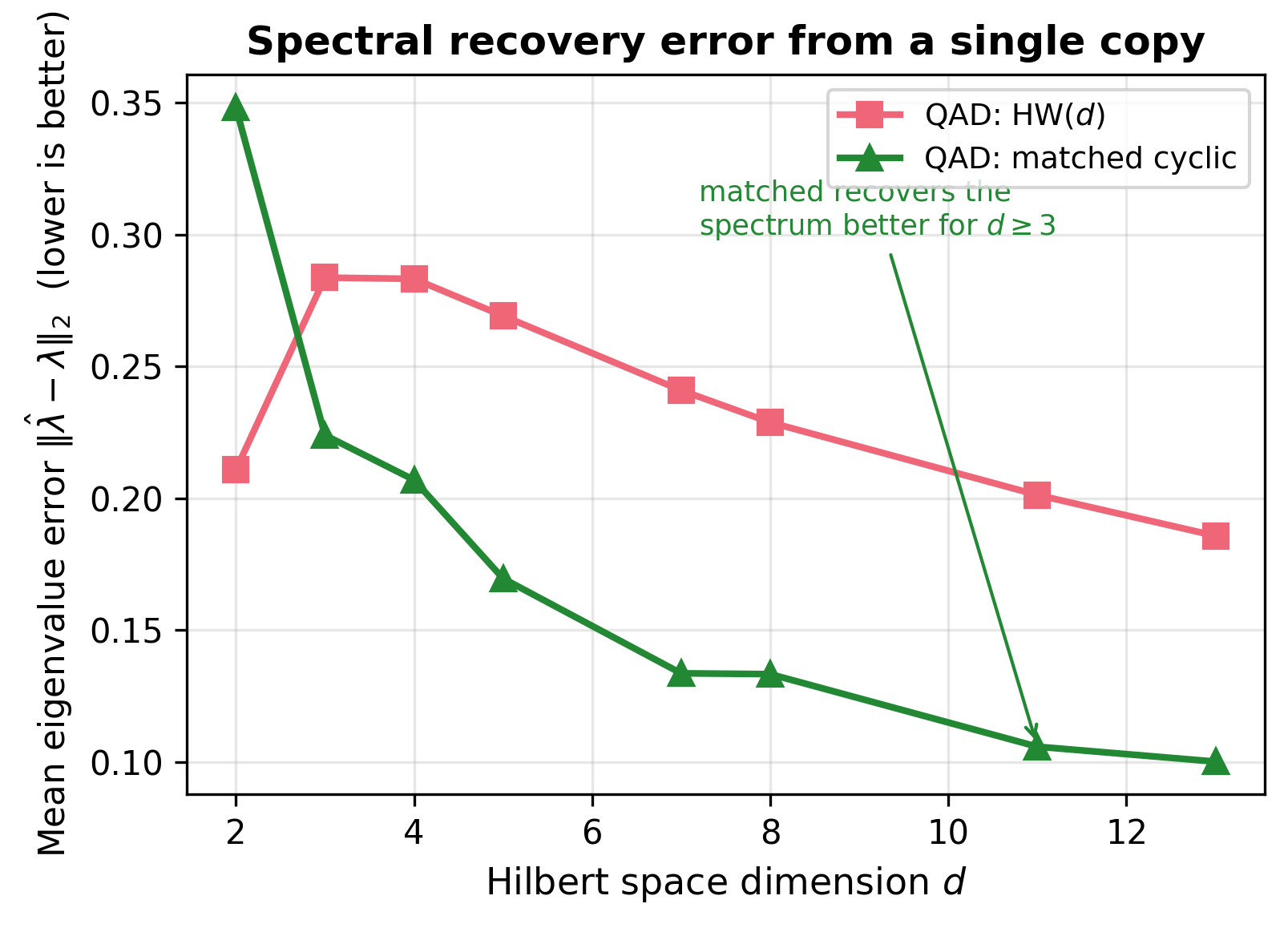}
\caption{Spectral recovery error (eigenvalue $\ell_2$ distance, lower is better) for the HW and matched group QAD estimators on the random-state ensemble (200 states per dimension, purity $\approx 0.7$, seed 42).  Despite its higher fidelity (Fig.~\ref{fig:fidelity}), the larger HW group recovers the eigenvalue spectrum \emph{less} accurately than the eigenbasis-matched cyclic group for $d \geq 3$ (e.g.\ $0.10$ vs.\ $0.19$ at $d = 13$): the HW estimate is biased toward $\mathbb{I}/d$, which flatters fidelity but contracts the recovered spectrum.  Group selection, not group size, governs spectral recovery.}
\label{fig:spectral}
\end{figure}

Three findings emerge from the simulation:

\textbf{Finding 1: QAD fidelity is dimension-independent} (Fig.~\ref{fig:fidelity}).  The HW($d$) group-averaged estimator maintains mean fidelity $F > 0.90$ from $d = 2$ through $d = 13$, while the standard rank-1 projector fidelity degrades as approximately $1/d$ (from 0.514 at $d = 2$ to 0.079 at $d = 13$).  This confirms the full-rank property of Theorem~\ref{thm:qad}(i): the $d^2$ group elements of $\mathrm{HW}(d)$ generate a spanning orbit that produces a full-rank estimator regardless of dimension.

\textbf{Finding 2: The fidelity ratio grows because the standard estimator collapses} (Table~\ref{tab:qudit}).  The HW/Standard ratio grows from $1.9\times$ at $d = 2$ to $11.3\times$ at $d = 13$ because the standard rank-1 estimator degrades as $\sim 1/d$ (closely tracking the $1/d$ guide in Fig.~\ref{fig:fidelity}) while the group-averaged fidelity is nearly constant in $d$.  This quantifies the conditioning advantage of a full-rank estimate from one outcome; it is not a reduction in the number of copies needed to reach a fixed error, and it does not extrapolate to a copy saving for multi-qubit systems.

\textbf{Finding 3: Fidelity and spectral recovery disagree, and the matched group recovers the spectrum better} (Fig.~\ref{fig:spectral}).  Although the Heisenberg-Weyl group attains the highest \emph{fidelity} (Fig.~\ref{fig:fidelity}), it does not recover the eigenvalue spectrum best: for $d \geq 3$ the matched cyclic group, which aligns with $\rho$'s eigenbasis, achieves the smaller eigenvalue error (e.g.\ $0.10$ versus $0.19$ at $d = 13$), and only at $d = 2$, where the full Pauli average collapses to $\mathbb{I}/2$, does HW score lower.  The two metrics point in opposite directions because the HW estimate is biased toward $\mathbb{I}/d$: that bias flatters fidelity against near-maximally-mixed targets while degrading spectral recovery.  This is direct evidence that single-copy fidelity is a forgiving metric and that group \emph{selection}, not group \emph{size}, is what recovers spectral structure.

\section{Implications for Quantum Error Correction}

Quantum error correction codes are defined by stabilizer groups, which are subgroups of the Pauli group that fix the code space.  In the QAD framework, a stabilizer code is a quantum state whose density matrix commutes with the stabilizer group: $[S, \rho_{\mathrm{code}}] = 0$ for all stabilizers $S$.

In the QAD framework, syndrome extraction (measuring which stabilizer is violated) can be viewed as a measurement with the stabilizer group, for which the commutativity condition is exactly satisfied (the stabilizer group is, by definition, the matched group for the code space).  We note, however, that syndrome extraction returns a sharp outcome, the specific violated stabilizer, a regime in which a change of measurement basis does not reduce the per-syndrome cost~\cite{thornton2026qbmg}; and group-averaging over the stabilizer group acts within the code space, where it leaves $\rho_{\mathrm{code}}$ fixed.  Whether the group-averaged estimator exposes the violated stabilizer is therefore not immediate.

We leave as a conjecture whether a single group-structured POVM generated by the stabilizer group can reduce the number of distinct measurement settings for simultaneous syndrome readout, relative to measuring each of the $O(n-k)$ stabilizers of an $[[n,k]]$ code separately.  We do not claim a reduction in the copies required.

\section{Structural Capacity and the von Neumann--Holevo Connection}

The QAD framework reveals a structural complement to the information-theoretic quantities introduced by von Neumann~\cite{vonneumann1932} and bounded by Holevo~\cite{holevo1973}.  The connection rests on the distinction between two entropic functionals of the density matrix eigenvalues, each governing a different aspect of quantum information.

\subsection{Von Neumann Entropy as Content; Structural Capacity as Structure}

The von Neumann entropy $S(\rho) = -\Tr(\rho \log \rho) = -\sum_k \lambda_k \log \lambda_k$ is the Shannon entropy ($H_1$, the Renyi entropy of order 1) applied to the eigenvalue distribution of $\rho$.  It quantifies the \emph{information content} of the quantum state: the maximum number of classical bits extractable per copy in the asymptotic limit.

The structural capacity of a density matrix, defined in analogy with the classical structural capacity~\cite{thornton2026ad}, is
\begin{equation}\label{eq:kappa_quantum}
\kappa(\rho) = 1 + \frac{(\Tr\,\rho)^2}{\Tr(\rho^2)} = 1 + \frac{1}{\Tr(\rho^2)},
\end{equation}
since $\Tr\,\rho = 1$.  The quantity $1/\Tr(\rho^2) = \exp(H_2(\rho))$ is the exponential of the Renyi-2 entropy of $\rho$'s eigenvalues, and is also the participation ratio of the eigenvalue distribution.  The structural capacity quantifies the \emph{effective dimensionality} of the state: the number of eigenvalues that contribute meaningfully to the spectral structure (Fig.~\ref{fig:kappa_vs_purity}).  For a pure state, $\Tr(\rho^2) = 1$ and $\kappa = 2$ (minimal structure: one active eigenvalue).  For the maximally mixed state $\rho = \mathbb{I}_d/d$, $\Tr(\rho^2) = 1/d$ and $\kappa = 1 + d$ (maximal structure: all eigenvalues active).

The two quantities are governed by different Renyi orders of the same eigenvalue distribution (Fig.~\ref{fig:content_vs_structure}):
\begin{center}
\small
\renewcommand{\arraystretch}{1.15}
\begin{tabular}{@{}lll@{}}
\toprule
\textbf{Quantity} & \textbf{Renyi order} & \textbf{Measures} \\
\midrule
$S(\rho)$ (von Neumann) & $\alpha = 1$ & Information content \\
$\kappa(\rho)$ (structural capacity) & $\alpha = 2$ & Exploitable structure \\
\bottomrule
\end{tabular}
\end{center}
This is the quantum manifestation of the classical content-vs-structure duality~\cite{thornton2026ad}: Shannon's theory (order 1) characterizes how much information a source produces; AD's structural theory (order 2) characterizes how much structure a single observation reveals.

\begin{figure}[t]
\centering
\includegraphics[width=\columnwidth]{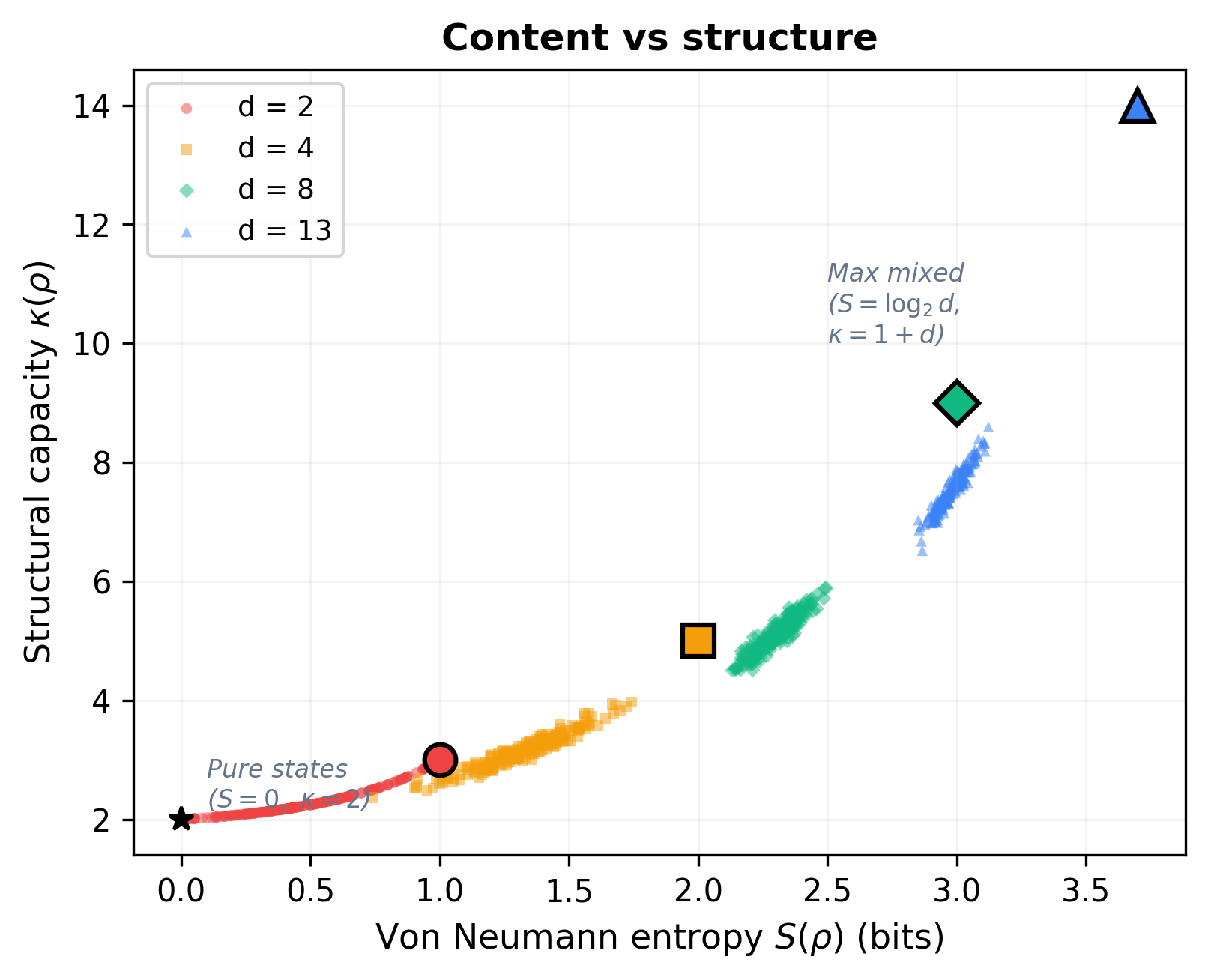}
\caption{Content vs structure for 200 random mixed states per dimension ($d = 2, 4, 8, 13$).  Von Neumann entropy $S(\rho)$ (Renyi-1, horizontal) measures information content; structural capacity $\kappa(\rho)$ (Renyi-2, vertical) measures exploitable structure.  The two quantities are correlated but not identical: states with equal content can have different structure, and vice versa.  Stars: pure states ($S = 0$, $\kappa = 2$).  Large markers with black edges: maximally mixed states ($S = \log_2 d$, $\kappa = 1 + d$).  Higher-dimensional systems access a larger range of both quantities.}
\label{fig:content_vs_structure}
\end{figure}

\subsection{The Holevo Bound and a Conjectured Structural Analogue}

Holevo's theorem~\cite{holevo1973} establishes the maximum classical information extractable from a quantum state ensemble.  For an ensemble $\{p_x, \rho_x\}$ with average state $\bar{\rho} = \sum_x p_x \rho_x$, the accessible information is bounded by the Holevo quantity $\chi = S(\bar{\rho}) - \sum_x p_x S(\rho_x)$.  This is a \emph{content} bound: it limits how many classical bits can be retrieved per quantum state, using von Neumann entropy ($H_1$).

The QAD framework suggests a parallel \emph{structural} quantity.  In the classical setting~\cite{thornton2026ad}, the structural capacity $\kappa$ counts the effective number of spectral parameters present in the data object.  Applying this to the density matrix:

\begin{remark}[Structural capacity as an upper bound]\label{prop:structural_holevo}
The participation ratio $1/\Tr(\rho^2)$, and hence $\kappa(\rho) = 1 + 1/\Tr(\rho^2)$, counts the effective number of spectral parameters \emph{present} in $\rho$, and so upper-bounds the number any estimator could resolve.  The matched-basis (eigenbasis) measurement is the most favorable for resolving them, but a single-copy estimate does not attain this count: the estimator~\eqref{eq:rho_G} is biased toward the symmetrized state, so single-copy recovery approaches a structural floor rather than a clean achievability bound.  The parallel with Holevo's content bound below is a conjectured structural analogue, not a proven measurement bound.
\end{remark}

The parallel is:

\begin{center}
\small
\renewcommand{\arraystretch}{1.15}
\begin{tabular}{@{}lll@{}}
\toprule
& \textbf{Holevo bound} & \textbf{Structural analogue} \\
\midrule
Limits & Extractable content & Exploitable structure \\
Functional & $S(\rho)$ (von Neumann) & $\kappa(\rho) = 1 + e^{H_2(\rho)}$ \\
Renyi order & $\alpha = 1$ & $\alpha = 2$ \\
Favored by & HSW coding & Matched-group POVM \\
\bottomrule
\end{tabular}
\end{center}

Holevo's bound limits content extraction regardless of the measurement strategy.  The structural analogue is weaker: $\kappa$ counts the structure present in $\rho$, an upper bound on what any group could resolve, but we do not establish a matching achievability statement for single-copy measurement.  A pure state has $\kappa = 2$ (minimal recoverable structure) and $S = 0$ (no extractable content).  The maximally mixed state has $\kappa = 1 + d$ (maximal recoverable structure) and $S = \log d$ (maximal extractable content).  The two quantities are complementary: a state can have high content but low structure (a nearly pure state with many off-diagonal coherences), or high structure but low content (a highly mixed state with a rich eigenvalue spectrum but little surprise in each measurement).

\begin{figure}[t]
\centering
\includegraphics[width=\columnwidth]{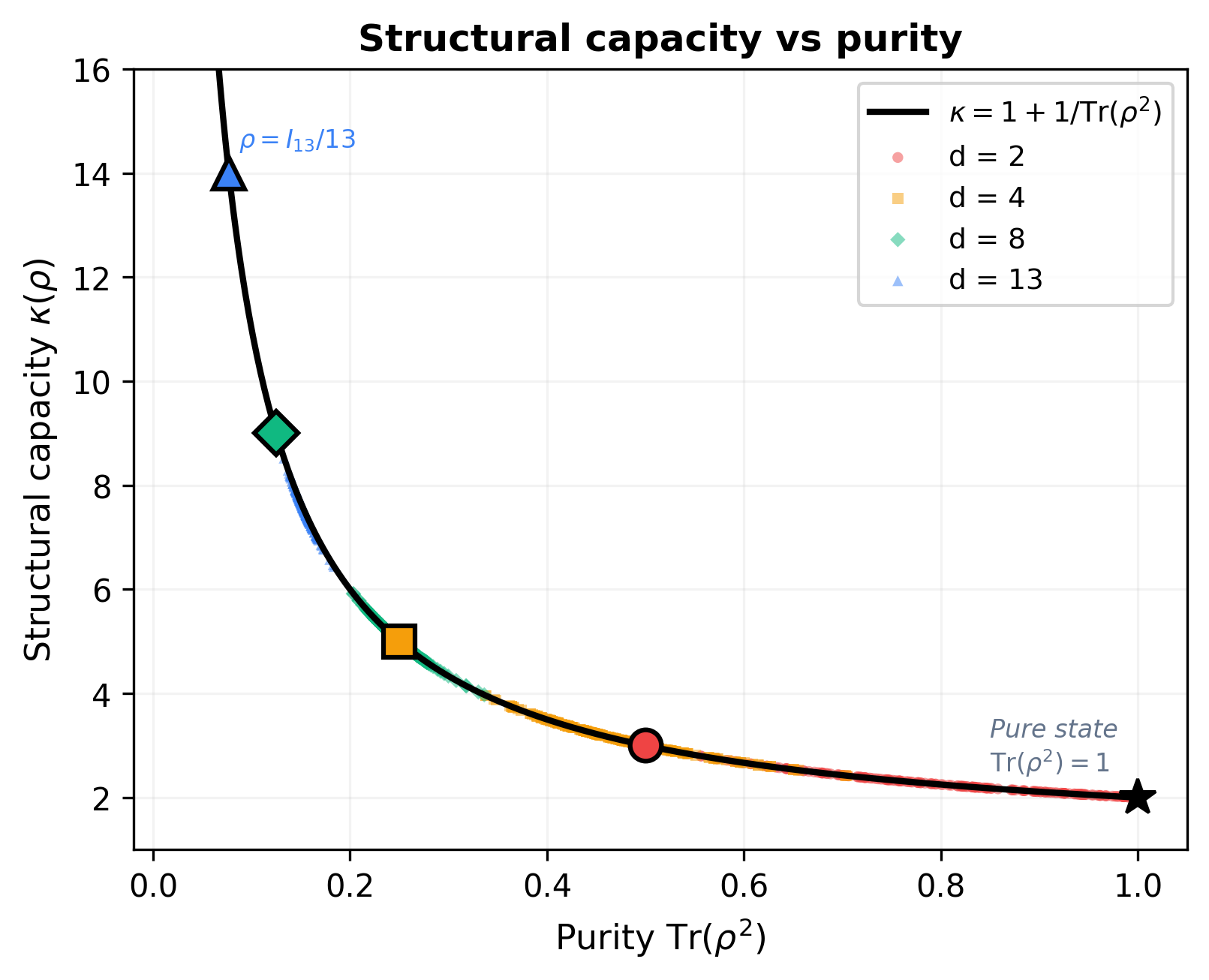}
\caption{Structural capacity $\kappa(\rho) = 1 + 1/\Tr(\rho^2)$ versus purity for random mixed states across four Hilbert space dimensions.  The black curve is the universal relationship $\kappa = 1 + 1/\text{purity}$.  All states from all dimensions lie exactly on this curve, confirming that $\kappa$ is determined solely by purity (equivalently, Renyi-2 entropy).  Pure states ($\Tr(\rho^2) = 1$) have $\kappa = 2$ regardless of $d$.  Maximally mixed states ($\Tr(\rho^2) = 1/d$, large markers with black edges) have $\kappa = 1 + d$.  Higher-dimensional systems access lower purity (more mixed states) and therefore higher structural capacity.}
\label{fig:kappa_vs_purity}
\end{figure}

\subsection{Non-Orthogonal States and the Measurement Basis Problem}

Von Neumann established that non-orthogonal quantum states permit richer encodings than orthogonal states, but Holevo's theorem limits how much of that richness is retrievable.  In the QAD framework, the measurement basis selection problem, which basis should one measure in to learn the most about $\rho$ from a single copy, is precisely the quantum group matching problem.  Non-orthogonal measurement bases (such as SIC-POVMs, which consist of $d^2$ non-orthogonal vectors) provide more algebraic views of the state than any orthogonal basis ($d$ vectors), at the cost of statistical overlap between outcomes.  The Heisenberg-Weyl group generates the maximal number of non-orthogonal views ($d^2$) from a single fiducial, which is why SIC-POVMs achieve the highest single-copy fidelity in the QAD simulations.

The set of mutually unbiased bases is the quantum group catalog: each MUB is a ``group'' for probing $\rho$, and the question ``which MUB reveals the most structure?'' is $G^* = \arg\min_G \delta_Q(G, \rho)$ in quantum language.  This connects QAD to the active research area of optimal quantum state estimation~\cite{haah2017sample}, providing a group-theoretic criterion for adaptive POVM selection that complements the information-theoretic criteria in the existing literature.  The group-selection criterion here, $G^* = \arg\min_G \delta_Q$, is unsupervised, operates on the covariance commutant, and uses a single observation; it is distinct from supervised group discovery over the regular representation that learns from a labeled dataset~\cite{hoyos2026gat}.

\section{Discussion}

\subsection{Relation to Existing Quantum Tomography}

Standard tomography~\cite{haah2017sample} requires $O(d^2/\varepsilon^2)$ copies.  Compressed sensing tomography~\cite{gross2010quantum} reduces this to $O(rd\log^2 d)$ for rank-$r$ states by exploiting low-rank structure.  The QAD approach exploits \emph{algebraic structure} rather than rank structure: when $\rho$ exactly commutes with a group $G$ having $r$ irreducible representations, only $r$ isotypic spectral parameters remain free, and they are estimated from $O(r/\varepsilon^2)$ copies rather than $O(d^2/\varepsilon^2)$.  This is the one regime in which QAD yields a genuine copy reduction, and it is conditional on the symmetry being present; for generic $\rho$ with no exact symmetry the parameter count, and hence the copy cost, is unchanged, and the benefit is the single informationally complete measurement setting.  These exploit orthogonal types of structure (algebraic symmetry vs.\ low rank), just as classical AD and compressed sensing exploit orthogonal types of signal structure~\cite{thornton2026ad}.

\subsection{Adaptive POVM via the Double Commutator}

The double-commutator GEVP~\cite{thornton2026dc} provides a polynomial-time algorithm for optimal POVM selection (Corollary~\ref{cor:adaptive}).  This enables a two-stage adaptive tomography protocol: (1)~use a small number of copies with a generic POVM (e.g., SIC) to obtain a coarse $\hat{\rho}$; (2)~solve the double-commutator GEVP to identify the optimal measurement group for $\hat{\rho}$; (3)~use the remaining copies with the optimized group-structured POVM.  This is the quantum analog of the classical adaptive group selection strategy described in~\cite{thornton2026ad}.

\subsection{Independence from Hardware Implementations}

The QAD Theorem is a mathematical result about the relationship between group structure and density matrix estimation.  It is independent of any particular physical implementation (photonic, superconducting, trapped-ion, etc.) and applies to any quantum system whose measurements can be described by group-structured POVMs.  The connection to specific hardware, such as reconfigurable photonic processors that can implement arbitrary unitary transforms~\cite{reck1994experimental}, is the subject of separate work.  Of particular interest for future quantum communication applications are frequency-bin encoded photonic qudits, where a parallel architecture of Mach-Zehnder modulators on thin-film lithium niobate (TFLN) enables programmable unitary operations on high-dimensional photon states at room temperature~\cite{macfarlane2025freq}.  Because photons maintain quantum coherence during propagation (``flying qubits''), such architectures could support transmission of group-structured quantum states between spatially separated nodes, enabling networked QAD protocols.

\subsection{Information Structure and Quantum Measurement}

The General Algebraic Averaging Theorem~\cite{thornton2026gaat} establishes that algebraic diversity exploits the \emph{structure} of information (the representation-theoretic symmetry of the data object) rather than the \emph{content} (the numerical values), complementing Shannon's information theory which characterizes content independently of structure.  The density matrix $\rho$ is a rank-2 tensor with $d^2 - 1$ parameters organized in a $d \times d$ Hermitian matrix with row-column index structure.  The Heisenberg-Weyl group acts on both indices simultaneously via the adjoint action $\rho \mapsto U_g \rho U_g^\dagger$, producing $d^2$ algebraically distinct views.  The effective dimension for this action is $d_{\mathrm{eff}} = d^2$, which governs the informational completeness of the measurement and the conditioning of the single-copy estimator, that is, the number of distinct measurement settings collapsed into one, rather than a reduction in the copies needed per parameter.

Standard tomography measures in a fixed basis, which is the trivial group applied to a rank-2 data object: a rank-0 measurement strategy applied to a rank-2 structure.  Building the rank-2 structure this way requires $O(d^2)$ distinct measurement settings, accumulated one projector at a time.  The QAD approach uses a single rank-matched measurement (a group whose order matches the parameter count of the density matrix) that is informationally complete, recovering the rank-2 structure from one setting.  What the group structure removes is the multiplicity of settings; the number of copies needed to drive each of the $d^2 - 1$ parameters to a fixed error remains governed by the sampling law, except in the symmetric case above, where the free-parameter count itself drops.

\section{Conclusion}

The QAD Theorem establishes that the algebraic diversity principle, generalizing the trivial-group measurement that standard tomography performs implicitly, extends from classical covariance estimation to quantum state estimation.  The Classical-Quantum Duality Map provides the formal bridge, and the Optimality Inheritance Theorem ensures that classical group selection results (including the polynomial-time double-commutator algorithm) transfer to quantum settings.  The identification of SIC-POVMs and MUBs as instances of algebraic diversity with the Heisenberg-Weyl and Clifford groups reveals a structural parallel between the classical transform hierarchy (DFT/DCT/KLT) and the quantum measurement hierarchy (SIC/MUB/general POVM) that has not been previously recognized.

The qubit example demonstrates that the QAD estimator achieves full-rank density matrix estimation with fidelity 0.99 from a single measurement outcome, compared to fidelity 0.80 for the standard single-outcome estimate on the same state.  Monte Carlo simulations on qudits of dimension $d = 2$ through $d = 13$ confirm that the Heisenberg-Weyl QAD estimator maintains fidelity above 0.90 from a single copy while the standard single-outcome fidelity degrades as $\sim 1/d$.  The growing ratio reflects the collapse of the rank-1 standard estimator, not a reduction in the copies required per parameter: the single-copy group-averaged estimator is biased toward the symmetrized state and trades sampling variance for a full-rank, well-conditioned structure estimate.  What the group structure removes is the multiplicity of measurement settings, replacing the $O(d^2)$ separate bases of standard tomography with one informationally complete group-structured POVM; the per-parameter copy cost continues to obey the sampling law.  A genuine reduction in copies arises only when $\rho$ carries an exact symmetry, in which case the free-parameter count drops from $d^2 - 1$ to the number of isotypic components.  The structural capacity $\kappa(\rho) = 1 + 1/\Tr(\rho^2)$, governed by the Renyi-2 entropy of $\rho$'s eigenvalue distribution, counts the effective spectral structure present in the state; a conjectured structural analogue of Holevo's content bound positions QAD as a structural complement to content-based measurement theory, with von Neumann entropy ($H_1$) characterizing how many bits a state conveys and $\kappa$ ($H_2$) characterizing the spectral organization available to a single observation.  Establishing whether this analogue holds as a measurement bound, and characterizing the bias floor of the single-copy estimator, are the natural next steps.

\end{document}